\documentclass[a4paper,12pt]{article} 
\usepackage{amsmath,amsthm,amssymb}
\newcommand{\bm}[1]{\mbox{\boldmath$#1$}}
\newcommand{\be}{\begin{equation}}
\newcommand{\ee}{\end{equation}}
\newcommand{\bea}{\begin{eqnarray}}
\newcommand{\eea}{\end{eqnarray}}
\newcommand{\non}{\nonumber}
\newtheorem{df}{Definition}
\newtheorem{th1}[df]{Theorem}
\newtheorem{lem}[df]{Lemma}
\newtheorem{conj}[df]{Conjecture}
\newtheorem{prop}[df]{Proposition}
\newtheorem{cor}[df]{Corollary}

\newcommand{\al}{\alpha}

\begin{document}
\title{Extension of a Borel subalgebra symmetry 
into the $sl_2$ loop algebra symmetry 
for the twisted XXZ spin chain  
at roots of unity and the Onsager algebra   
\footnote{Talk given at the workshop RAQIS, LAPTH, 
Annecy, France, September 11-14, 2007}} 
\author{Tetsuo  Deguchi}
\vskip 24pt 
%\vskip12pt 

\date{}
\maketitle

\begin{center}  
Department of Physics, Ochanomizu University \\
2-1-1 Ohtsuka, Bunkyo-ku, Tokyo 112-8610, Japan 
\end{center}

\abstract{
We discuss a conjecture 
that the twisted transfer matrix of the six-vertex model 
at roots of unity with some discrete twist angles 
should have the sl(2) loop algebra symmetry. 
As an evidence of this conjecture, 
we show the following mathematical result on a  
 subalgebra of the sl(2) loop algebra, which we  call 
a Borel subalgebra:  
 any given finite-dimensional highest weight representation 
of the Borel subalgebra is extended into that of 
the sl(2) loop algebra, if the parameters associated with it are nonzero. 
Thus, if operators  commuting or anti-commuting with 
 the twisted transfer matrix of the six-vertex model at roots of unity 
generate the Borel subalgebra, 
then they also generate the sl(2) loop algebra. 
The result should be useful for studying  
the connection of the sl(2) loop algebra symmetry  
to the Onsager algebra symmetry of 
the superintegrable chiral Potts model. 
}

%%%%%%%%%%%%%%%%%%%%%%%%%%%%%%%%%%%%%%%%%%%%%%%%%%%%%%%%%%%%%
\section{Introduction} 

Spectral properties of 
the XXZ spin chain under the twisted boundary conditions 
have attracted much attention in mathematical physics 
and condensed matter physics \cite{Barber,Alcaraz,SS,ACTWu,Fowler}. 
The XXZ Hamiltonian on a ring of $L$ sites is given by  
\be  
{\cal H}_{\rm XXZ} =  J \sum_{j=1}^{L} \left(\sigma_j^X \sigma_{j+1}^X +
 \sigma_j^Y \sigma_{j+1}^Y + \Delta \sigma_j^Z \sigma_{j+1}^Z  \right) \, . 
\label{hxxz}
\ee
where $\sigma_j^{\alpha}$ ($\alpha=X,Y,Z$) are the Pauli matrices  
defined on the $j$th site, and  they satisfy the   
twisted boundary conditions:   
\be  
\sigma_{L+1}^{\pm} = \exp(\pm i \phi) \sigma_{1}^{\pm}  \, , 
 \qquad \sigma^{Z}_{L+1} = \sigma^{Z}_{1} . 
\label{eq:TBC}
\ee
We call the parameter $\phi$ the twist angle. When $\phi=0$,  
conditions (\ref{eq:TBC}) reduces to the periodic boundary conditions.  
We define parameter $q$ by $\Delta=(q+q^{-1})/2$. 
We also introduce twist parameter $\varphi$ by 
\be  
q^{2 \varphi} = \exp(i \phi) \, .
\ee

It has been shown that when $q$ is a root of unity 
the XXZ spin chain under the periodic boundary conditions 
commutes with the $sl_2$ loop algebra, $U(L(sl_2))$ \cite{DFM}. 
(See also, \cite{FM1,FM2,Odyssey,criterion,RegularXXZ}.) 
Through the similar derivation in terms 
of the Temperley-Lieb algebra as given in \cite{DFM}, 
it was shown that the twisted XXZ spin chain at roots of unity 
commutes with the $sl_2$ loop algebra for $\phi=\pi$, i.e. 
under the anti-periodic boundary conditions \cite{twisted}.   
It was also shown that when $q$ is a root of unity such as $q^{2N}=1$ 
and  $\varphi$ is an integer, 
there exist some operators commuting or anti-commuting 
 with the twisted transfer matrix of the six-vertex model \cite{twisted}.  
Furthermore, it was pointed out by Korff that in some sectors 
such operators generate a subalgebra  $U({\cal B}_0)$ 
of the $sl_2$ loop algebra $U(L(sl_2))$,  
which we call a Borel subalgebra \cite{Korff}.  
Let $x_m^{\pm}$ and $h_n$ for $m,n \in {\bm Z}$ be 
the generators of the $sl_2$ loop algebra $U(L(sl_2))$. 
Then, the Borel subalgebra $U({\cal B}_0)$ is 
generated by the following operators:     
$x_k^{+}, h_k$ for $k = 0, 1, \ldots,$ and    
$x_k^{-}$ for $k = 1, 2, \ldots $.

In the paper we show a mathematical result that  
every highest weight representation of the Borel subalgebra $U({\cal B}_0)$ 
is extended into that of the $sl_2$ loop algebra 
if the parameters associated with the representation are nonzero. 
It follows from the result that if the twisted 
transfer matrix has the Borel subalgebra symmetry, then 
it has also the $sl_2$ loop algebra symmetry.  
 We thus give a conjecture that the $sl_2$ loop algebra 
is generated by the operators constructed in \cite{twisted}
which commute or anti-commute 
 with the twisted transfer matrix of the six-vertex model 
at roots of unity. 
Here we note that Benkart and Terwilliger have shown 
that the action of $U({\cal B}_0)$ on 
a finite-dimensional irreducible 
$U({\cal B}_0)$ module extends uniquely to an action of 
$U_q(L(sl_2))$ on it \cite{Benkart}. 
The mathematical result in the paper 
is new for reducible highest weight representations 
of $U({\cal B}_0)$.  
We also discuss construction of generators of the 
Onsager algebra from a highest weight representation of 
the $sl_2$ loop algebra. The result should be 
useful for investigating the connection of  
the $sl_2$ loop algebra   
to the Onsager algebra symmetry of 
the super-integrable chiral Potts model 
\cite{ND}. Quite recently in an independent research \cite{Helen}, 
eigenvectors of the superintegrable model 
associated with the superintegrable chiral Potts model  
have been studied by making use of the $sl_2$ loop algebra symmetry 
of some XXZ spin chain. They should be 
 closely related to Ref. \cite{ND}, and   
some results of the present paper should also be relevant.  

The content of the paper consists of the following: 
In section 2, we review the infinite-dimensional symmetries of 
the twisted transfer matrix of the six-vertex model at roots of unity. 
In particular, we review operators commuting or anti-commuting with 
the twisted transfer matrix at roots of unity.   
In some sectors they generate the Borel subalgebra. 
In section 3, we show that any given highest weight representation of 
the Borel subalgebra is extended to that of the $sl_2$ loop algebra 
if the associated parameters are nonzero. 
In section 4, we summarize some results on 
the infinite-dimensional symmetry of the twisted transfer matrix of the 
six-vertex model at roots of unity, and then suggest a conjecture 
that the twisted transfer matrix of the 
six-vertex model at roots of unity should have the $sl_2$ loop algebra 
symmetry. In section 5 , we give a method for constructing 
a representation of the Onsager algebra from a finite-dimensional 
highest weight representation of the $sl_2$ algebra. 

%

%%%%%%%%%%%%%%%%%%%%%%%%%%%%%%%%%%%%%%%%%%%%%%%%%%%%%%%%%%%%%%%
%
\section{Infinite dimensional symmetry of the twisted XXZ spin chain}  
\subsection{Definition of the twisted transfer matrix}

In order to formulate the twisted transfer matrix of the six-vertex model, 
we review some formulas of the algebraic Bethe ansatz. 
%Here we recall that  parameter $2 \eta$ is related to $q$  by $q=\exp(2\eta)$.
The $R$ matrix of the XXZ spin chain is defined by  
\be
R(z-w)  
= \left(
\begin{array}{cccc} 
f(w-z) &   0 & 0 & 0 \\
0 &   g(w-z) & 1 & 0 \\
0 &  1 & g(w-z) & 0 \\
0 &   0 & 0 & f(w-z)
\end{array} 
\right) 
\ee
where $f(z-w)$ and $g(z-w)$ are given by 
\be 
f(z-w)= {\frac {\sinh(z-w-2 \eta)} {\sinh(z-w)}} \, , \quad
g(z-w)= {\frac {\sinh(-2 \eta)} {\sinh(z-w)}} \, . 
\label{fg}
\ee
%Hereafter we shall write $f(w_j-w_k)$ simply as $f_{j,k}$ 
%for some parameters $w_j$.    
% 
We introduce $L$ operators for the XXZ spin chain  
\be 
 L_n(z) 
= \left(
\begin{array}{cc}  
L_n(z)^1_1  &  L_n(z)^1_2 \\
 L_n(z)^2_1  & L_n(z)^2_2  
\end{array} 
\right)  
=  \left(
\begin{array}{cc}  
\sinh \left( z \, I_n + \eta \sigma_n^Z \right) 
& \sinh 2 \eta \, \sigma_n^{-} \\
\sinh 2 \eta \, \sigma_n^{+}  
& \sinh \left( z \, I_n - \eta \sigma_n^Z \right) 
\end{array} 
\right)  
\ee
Here $I_n$ and $\sigma_n^a$ ($n=1, \ldots, L$) are 
acting on the $n$th vector space $V_n$. 
We recall that  $\sigma^{\pm}$ denote 
$\sigma^{+}= E_{12}$ and $\sigma^{-} = E_{21}$,    
and $\sigma^X, \sigma^Y, \sigma^Z$ the Pauli matrices.   
In terms of the $R$ matrix and $L$ operators, 
the Yang-Baxter equation is expressed as 
\be 
R(z-w) \left( L_n(z) \otimes L_n(w) \right) 
 = \left( L_n(w) \otimes L_n(z) \right) R(z-w) 
\label{RLL} 
\ee
We define the  monodromy matrix $T(z)$ by 
$T(z) = L_L(z) \cdots L_2(z) L_1(z) \, $.   
The monodromy matrix satisfies the Yang-Baxter equations 
\be 
R(z-w) \, \left( T(z; \{ \xi_n \}) \otimes T(w; \{ \xi_n \})  \right) 
= \left( T(w; \{\xi_n \}) \otimes T(z; \{\xi_n \}) \right) 
\, R(z-w) \label{ybr}
\ee
Let us denote the matrix elements of $T(z)$ as follows:  
\be 
T(z) = 
\left( \begin{array}{cc}
A(z) & B(z) \\
C(z) & D(z) 
\end{array} 
\right)
\ee
The twisted transfer matrix $\tau_{6V}(z; \varphi)$ is defined by 
\be 
\tau_{6V}(z; \varphi) 
= {\rm tr} \left( \, q^{\varphi \sigma_0^z} T(z) \right)= 
q^{\varphi} A(z) + q^{-\varphi} D(z).
\ee
The twisted Hamiltonian is given by the following logarithmic derivative: 
\bea
& &  {\sinh 2 \eta}  \times  
\, {\frac d {dz} } \log \tau(z; \varphi)|_{z=\eta}  
\non \\
& & =   \sum_{j=1}^{L-1} 
\left( 2 \sigma_{j}^{+} \sigma_{j+1}^{-} +  
2 \sigma_{j}^{-} \sigma_{j+1}^{+} + \cosh 2 \eta  
\sigma_{j}^{Z} \sigma_{j+1}^{Z} \right) 
\non \\ 
&  & \qquad + \,  
  q^{- 2\varphi} \, 2 \sigma_{L}^{+} \sigma_{1}^{-} +  
q^{2 \varphi} \, 2 \sigma_{L}^{-} \sigma_{1}^{+} + \cosh 2 \eta  
\sigma_{L}^{Z} \sigma_{1}^{Z}  +  L \cosh 2 \eta  \non \\
&  & \qquad = \,  
{\cal H}_{XXZ}(\phi)/J + L \Delta \, .  \non 
%\label{twistedXXZ}
\eea

\subsection{Roots of unity conditions}

Let us formulate roots of unity conditions explicitly as follows
\cite{twisted,RegularXXZ}.   
\begin{df}[Roots of unity conditions]
We say that $q_0$ is a root of unity with $q_0^{2N}=1$, if one of the three 
conditions hold:  
(1) $N$ is odd and $q_0$ is a primitive $N$th root of unity, 
i.e. $q_0^N=1$;  
(2) $N$ is odd and $q_0$ is a primitive $2N$th root of unity 
, i.e. $q_0^N=-1$; 
(3) $N$ is even and $q_0$ is a primitive $2N$th root of unity 
, i.e. $q_0^N=-1$.  
%Furthermore, we call the cases (1) and (3)  type I, the case (2)  type II.   
\label{df:roots}
\end{df} 

Let us denote by  $S^Z \pm \varphi$ either $S^Z + \varphi$ or 
$S^Z - \varphi$.  
We now consider the condition of $q_0^{2 S^Z \pm 2 \varphi}=1$.  
The values of $S^Z$ and $\varphi$ are given by integers or half-integers 
under the twisted boundary conditions.   
\par \noindent 
(1) When $N$ is odd and $q_0^N=1$, we have 
$q_0^{2 S^Z \pm 2 \varphi}=1$ 
if and only if $S^Z \pm \varphi \equiv 0$ (mod $N$) 
or $S^Z \pm \varphi \equiv N/2$ (mod $N$). 
When $S^Z \pm \varphi \equiv 0$ (mod $N$),  
$\varphi$ is given by an integer  for even $L$, 
and a half-integer for odd $L$.   
When $S^Z \pm \varphi \equiv N/2$ (mod $N$),  
$\varphi$ is given by a half-integer for even $L$, 
and an integer  for odd $L$.   
\par \noindent
(2) When $N$ is odd and $q_0^N=-1$, we have 
$q^{2 S^Z \pm 2 \varphi}=1$ 
if and only if $S^Z \pm \varphi \equiv 0$ (mod $N$).  
$\varphi$ is given by an integer for even $L$, 
and a half-integer  for odd $L$.   
\par \noindent
(3) When $N$ is even and $q_0^N=-1$, we have 
$q^{2 S^Z \pm 2 \varphi}=1$ 
if and only if $S^Z \pm \varphi \equiv 0$ (mod $N$).  
$\varphi$ is given by an integer for even $L$, 
and a half-integer for odd $L$.   

Here we note that 
if the number of lattice sites $L$ is given by an even integer, 
then $S^Z$ takes integral values, while if $L$ is odd, 
$S^Z$ takes half-integral values. 

\subsection{Operators commuting with the twisted XXZ Hamiltonian}

We now formulate operators commuting or anti-commuting 
with the twisted transfer matrix of the six-vertex model at roots 
of unity \cite{twisted}.  
We introduce operators $S_j^{\pm}$ and $T_j^{\pm}$ by 
\bea 
S_j^{\pm} &= &  
q^{\sigma^Z/2} \otimes \cdots \otimes q^{\sigma^Z/2} \otimes 
\sigma_j^{\pm} \otimes q^{-\sigma^Z/2} \otimes \cdots \otimes q^{-\sigma^Z/2} 
\, , \non \\  
T_j^{\pm} & = & 
q^{-\sigma^Z/2} \otimes \cdots \otimes q^{-\sigma^Z/2} \otimes 
\sigma_j^{\pm} \otimes q^{\sigma^Z/2} \otimes \cdots \otimes q^{\sigma^Z/2}
\quad (j=1, 2, \ldots, L) . 
\eea
We define $S^{\pm}$ and $T^{\pm}$ by 
\be 
S^{\pm} = \sum_{j=1}^L S_j^{\pm} \, , 
\quad  T^{\pm} = \sum_{j=1}^L T_j^{\pm} \, . 
\ee 
They are generators of the affine quantum group $U_q({\hat {sl}}(2))$. 

Let us introduce the $q$-integer $[n]$ and the $q$-factorial $[m]!$,  
respectively, by the following:  
\be 
[n]= {\frac {q^n-q^{-n}} {q-q^{-1}}} \, , \quad 
[m]! = \prod_{k=1}^{m} [k] \, . 
\ee
It is easy to show 
\bea 
(S^{\pm})^m & = & q^{\pm m(m-1)/2} [m]! \, 
\sum_{1 \le i_1 < \cdots < i_m \le L} 
S_{i_1}^{\pm} \cdots S_{i_m}^{\pm} \, , \non \\ 
(T^{\pm})^m & = & q^{\mp m(m-1)/2} [m]! \, 
\sum_{1 \le i_1 < \cdots < i_m \le L} 
T_{i_1}^{\pm} \cdots T_{i_m}^{\pm} \, . 
\eea
The symbols $S^{\pm(N)}$ and $T^{\pm(N)}$ are defined in Ref. \cite{DFM} by 
\be 
 S^{\pm(N)} = \lim_{q \rightarrow q_0} {\frac {(S^{\pm})^N} {[N]!}} \, , 
\quad 
 T^{\pm(N)} = \lim_{q \rightarrow q_0} {\frac {(T^{\pm})^N} {[N]!}} \, . 
\ee
Here we define $(S^{\pm})^{(m)}$ and $(T^{\pm})^{(m)}$ 
for all positive integers $m$ by 
\be 
 (S^{\pm})^{(m)} = \lim_{q \rightarrow q_0} {\frac {(S^{\pm})^m} {[m]!}} \, , 
\quad 
 (T^{\pm})^{(m)} = \lim_{q \rightarrow q_0} {\frac {(T^{\pm})^m} {[m]!}} \, . 
\ee
Explicitly, we have $(S^{\pm})^{(m)}$ for any positive integer $m$ as follows. 
\begin{eqnarray}
(S^{\pm})^{(m)}  
&=&  \sum_{1 \le j_1 < \cdots < j_m \le L}
q_0^{{m \over 2 } \sigma^Z} \otimes \cdots \otimes q_0^{{m \over 2} \sigma^Z}
\otimes \sigma_{j_1}^{\pm} \otimes
q_0^{{(m-2) \over 2} \sigma^Z} \otimes  \cdots \otimes 
q_0^{{(m-2) \over 2} \sigma^Z} \nonumber \\
 & & \otimes \sigma_{j_2}^{\pm} \otimes 
 q_0^{{(m-4) \over 2} \sigma^Z} \otimes
\cdots
\otimes \sigma^{\pm}_{j_m} \otimes q_0^{-{m \over 2} \sigma^Z} 
\otimes \cdots \otimes q_0^{-{m \over 2} \sigma^Z}  \, . 
\label{sn}
\end{eqnarray}

Let $m$ and $n$ be integers such that 
$|m-n|= kN$ for some integer $k$.   
When $q_0$ is a root of unity with $q^{2N}=1$, we have the following. 
\begin{itemize} 
\item[(1)]
In the sectors of $S^Z \equiv -\varphi+n ({\rm mod} N)$, we have 
$$
(S^{+})^{(m)} (T^{-})^{(n)} \tau(z; \varphi) 
= q_0^{m-n} \, \tau(z; \varphi) (S^{+})^{(m)} (T^{-})^{(n)} 
$$
\item[(2)]
In the sectors of $S^Z \equiv -\varphi-n ({\rm mod} N)$, we have 
$$
(T^{-})^{(m)} (S^{+})^{(n)} \tau(z; \varphi) 
= q_0^{m-n} \, \tau(z; \varphi) (T^{-})^{(m)} (S^{+})^{(n)} 
$$
\item[(3)]
In the sectors of $S^Z \equiv \varphi-n ({\rm mod} N)$, we have 
$$
(S^{-})^{(m)} (T^{+})^{(n)} \tau(z; \varphi) 
= q_0^{m-n} \tau(z; \varphi) (S^{-})^{(m)} (T^{+})^{(n)} 
$$
\item[(4)]
In the sectors of $S^Z \equiv \varphi+n ({\rm mod} N)$, we have 
$$
(T^{+})^{(m)} (S^{-})^{(n)} \tau(z; \varphi) 
= q_0^{m-n} \tau(z; \varphi) (S^{-})^{(m)} (T^{+})^{(n)} 
$$
\end{itemize} 
Here we note that $q_0^N = \pm1$ when $q_0$ is a root of unity 
with $q_0^{2N}=1$. Thus we have $q_0^{m-n}=(\pm 1)^k$.  For simplicity, 
we have not considered the case when 
$N$ is odd with $q_0^N=1$ and $S^Z + \varphi \equiv N/2$ (mod $N$) 
or $S^Z - \varphi \equiv N/2$ (mod $N$).

\subsection{Examples}
For an illustration, we consider 
the case of a root of unity 
where $N=3$ ($q_0^{3}=1$)  and  $L$ is even.   
Some of the operators commuting or anti-commuting with the twisted 
transfer matrix are given as follows. 
\begin{itemize}
\item
[(1a)] $\varphi=0$ and $S^Z \equiv 0$ (mod $N$) 
$$
(S^{+})^{(3)} \, , \quad  (S^{-})^{(3)}\, , \quad 
  (T^{+})^{(3)}\, , \quad   (T^{-})^{(3)} 
$$ 
They generate the $sl_2$ loop algebra \cite{DFM}. 
\item
[(1b)] $\varphi=0$ and $S^Z \equiv 1$ (mod $N$): 
\bea 
& &  (S^{+})^{(4)} (T^{-})^{(1)} 
\, , \quad 
(T^{-})^{(5)} (S^{+})^{(2)} 
\, , \quad 
(S^{-})^{(5)} (T^{+})^{(2)} 
\, , \quad 
(T^{+})^{(4)} (S^{-})^{(1)} \, , \non \\ 
& & (S^{+})^{(1)} (T^{-})^{(4)}  
\, , \quad 
(T^{-})^{(2)} (S^{+})^{(5)} 
\, , \quad 
(S^{-})^{(2)} (T^{+})^{(5)} 
\, , \quad 
(T^{+})^{(1)} (S^{-})^{(4)} \, , \ldots \, .  \non
\eea
It is conjectured that they generate the $sl_2$ loop algebra 
\cite{DFM}.  
\item[(1c)] $\varphi=0$ and $S^Z \equiv 2$ (mod $N$): 
\bea 
& & (S^{+})^{(5)} (T^{-})^{(2)} 
\, , \quad 
(T^{-})^{(4)} (S^{+})^{(1)} 
\, , \quad 
(S^{-})^{(4)} (T^{+})^{(1)} 
\, , \quad 
(T^{+})^{(5)} (S^{-})^{(2)} \, , \non \\  
& & (S^{+})^{(2)} (T^{-})^{(5)} 
\, , \quad 
(T^{-})^{(1)} (S^{+})^{(4)} 
\, , \quad 
(S^{-})^{(1)} (T^{+})^{(4)} 
\, , \quad 
(T^{+})^{(2)} (S^{-})^{(5)} \, , \ldots \, . \non 
\eea 
It is conjectured that they should generate the $sl_2$ loop algebra 
\cite{DFM}. 
\end{itemize}

\begin{itemize}
\item[(2a)] $\varphi=1$ and $S^Z \equiv 0$ (mod $N$):   
\bea 
& & (S^{+})^{(4)} (T^{-})^{(1)} 
\, , \quad 
(T^{-})^{(5)} (S^{+})^{(2)} 
\, , \quad 
(S^{-})^{(5)} (T^{+})^{(2)} 
\, , \quad 
(T^{+})^{(4)} (S^{-})^{(1)} \, , \ldots \, , \non \\ 
& & (S^{+})^{(1)} (T^{-})^{(4)} 
\, , \quad 
(T^{-})^{(2)} (S^{+})^{(5)} 
\, , \quad 
(S^{-})^{(2)} (T^{+})^{(5)} 
\, , \quad 
(T^{+})^{(1)} (S^{-})^{(4)} \, , \ldots \, . \non 
\eea
%($U(Lsl_2)$ conj. )  
\item[(2b)] $\varphi=1$ and $S^Z \equiv 1$ (mod $N$):  
\bea 
& & (S^{+})^{(5)} (T^{-})^{(2)} 
\, , \quad 
(T^{-})^{(4)} (S^{+})^{(1)} 
\, , \quad 
(S^{-})^{(3)}  
\, , \quad (T^{+})^{(3)} \, , \ldots \, , \non \\ 
& & (S^{+})^{(2)} (T^{-})^{(5)} 
\, , \quad 
(T^{-})^{(1)} (S^{+})^{(4)} 
 \, , \ldots \, .  \non 
\eea
$(S^{-})^{(3)}$ and $(T^{+})^{(3)}$ generate 
a Borel subalgebra \cite{Korff}. 
% ($U(Lsl_2)$ conj. )  
\item[(2c)] $\varphi=1$ and $S^Z \equiv 2$ (mod $N$) 
%($U(Lsl_2)$ conj.)  
\bea 
& & (S^{+})^{(3)} 
\, , \quad 
(T^{-})^{(3)} 
\, , \quad 
(S^{-})^{(5)} (T^{+})^{(2)} 
\, , \quad 
(T^{+})^{(4)} (S^{-})^{(1)}  \, , \ldots \, , \non \\ 
& & \qquad \qquad  
  \quad 
(S^{-})^{(2)} (T^{+})^{(5)} 
\, , \quad 
(T^{+})^{(1)} (S^{-})^{(4)}  \, , \ldots \, . 
\eea
$(S^{+})^{(3)}$ and $(T^{-})^{(3)}$ generate 
a Borel subalgebra \cite{Korff}. 
\end{itemize}

\begin{itemize}
\item[(3a)] $\varphi=2$ and $S^Z \equiv 0$ (mod $N$):   
\bea 
& & (S^{+})^{(5)} (T^{-})^{(2)} 
\, , \quad 
(T^{-})^{(4)} (S^{+})^{(1)} 
\, , \quad 
(S^{-})^{(5)} (T^{+})^{(2)} 
\, , \quad 
(T^{+})^{(4)} (S^{-})^{(1)} \, , \ldots, , \non \\
&& (S^{+})^{(2)} (T^{-})^{(5)} 
\, , \quad 
(T^{-})^{(1)} (S^{+})^{(4)} 
\, , \quad 
(S^{-})^{(2)} (T^{+})^{(5)} 
\, , \quad 
(T^{+})^{(1)} (S^{-})^{(4)} \, , \ldots, . \non 
\eea
%($U(Lsl_2)$ conj. )  
\item[(3b)] $\varphi=1$ and $S^Z \equiv 1$ (mod $N$):  
\bea 
& & (S^{+})^{(3)} 
\, , \quad 
(T^{-})^{(3)} 
\, , \quad 
(S^{-})^{(4)} (T^{+})^{(1)} 
\, , \quad 
(T^{+})^{(5)} (S^{-})^{(2)} \, , \non \\ 
& &  
\quad \quad  
\quad \quad 
(S^{-})^{(1)} (T^{+})^{(4)} 
\, , \quad 
(T^{+})^{(2)} (S^{-})^{(5)} \, , \ldots, . \non 
\eea
$(S^{+})^{(3)}$ and $(T^{-})^{(3)}$ generate 
a Borel subalgebra \cite{Korff}. 
% ($U(Lsl_2)$ conj. )  
\item[(3c)] $\varphi=1$ and $S^Z \equiv 2$ (mod $N$): 
\bea 
& & (S^{+})^{(4)} (T^{-})^{(1)} 
\, , \quad 
(T^{-})^{(5)} (S^{+})^{(2)} 
\, , \quad 
(S^{-})^{(3)}  
\, , \quad (T^{+})^{(3)} \, , \ldots, , \non \\ 
& & (S^{+})^{(1)} (T^{-})^{(4)} 
\, , \quad 
(T^{-})^{(2)} (S^{+})^{(5)} 
\, ,  \ldots, . \non 
\eea 
$(S^{-})^{(3)}$ and $(T^{+})^{(3)}$ generate 
a Borel subalgebra \cite{Korff}. 
\end{itemize}

%%%%%%%%%%%%%%%%%%%%%%%%%%%%%%%%%%%%%%%%%%%%%%%%%%%%%%%%%%%%%
%
\section{Extension of the Borel subalgebra symmetry}
\subsection{Definition of the Borel subalgebra of $U(L(sl_2))$}
%
%%%%%%%%%%%%%%%%%%%%%%%%%%%%%%%%%%%%%%%%%%%%%%%%%%%%%%%%%%%%%%

We recall that the Borel subalgebra, $U({\cal B}_0)$, 
is generated by the following operators:     
$$
x_k^{+}, h_k \quad  {\rm for} \, k = 0, 1, \ldots, \quad   
{\rm and} \quad 
x_k^{-} \quad {\rm for} \, k = 1, 2, \dots .   
$$
They satisfy the defining relations given as follows: 
\bea 
{[} h_j, x_k^{+} {]} & = & 2 x_{j+k}^{+} \, , 
\quad {\rm for} \quad j, k \ge 0 , 
%\,{ \rm and} \,  k \ge 0,  
\non \\ 
{[} h_j, x_k^{-} {]} & = & (- 2) x_{j+k}^{-} \, , \quad {\rm for} 
\quad j \ge 0 \,{\rm and} \,  k \ge 1,  \non \\
{[} x^{+}_j, x_k^{-} {]} & = & \delta_{j,k} h_{j+k} \, , 
\quad {\rm for} \quad j \ge 0 \,{\rm and} \,  k \ge 1,  \non \\ 
{[} h_j, h_k {]} & =& 0 \, ,  \quad {\rm for} \quad 
j, k \ge 0 , 
%\,{\rm and} \,  k \ge 0,  
\non \\  
  {[} x_j^{+}, x_k^{+} {]} & = & 0  
\quad {\rm for} \quad j \ge 0 \,{\rm and} \,  k \ge 0,  \non \\ 
  {[} x_j^{-}, x_k^{-} {]} & = & 0  
\quad {\rm for} \quad j \ge 1 \,{\rm and} \,  k \ge 1 . \label{eq:dfrBorel}  
\eea

%%%%%%%%%%%%%%%%%%%%%%%%%%%%%%%%%%%%%%%%%%%%%%%%%%%%%%%%%%%%%
%
\subsection{Highest weight vectors and highest weight parameters} 
%
%%%%%%%%%%%%%%%%%%%%%%%%%%%%%%%%%%%%%%%%%%%%%%%%%%%%%%%%%%%%%%

Let us  define highest weight vectors of the Borel subalgebra 
$U({\cal B}_0)$.  

\begin{df} 
In a representation of $U({\cal B}_0)$, 
we call a vector $\Psi$ 
a highest weight vector if it is annihilated by all $x_k^{+}$'s, i.e.    
$x_k^{+} \Psi = 0$ for $k=0, 1, \ldots$, 
and is a simultaneous eigenvector of all $h_k$'s, i.e. 
$h_k \Psi = d_k \Psi$ for $k=0, 1, \ldots .$ 
We call the set of eigenvalues $d_k$ the highest weight of $\Psi$. 
We call the representation generated by a highest weight 
vector $\Psi$, the  highest weight representation of $\Psi$. We 
denote it by $U({\cal B}_0) \Psi$. 
\end{df}

%%%%%%%%%%%%%%%%%%%%%%%%%%%%%%%%%%%

\begin{df} 
Let $\Psi$ be a highest weight vector of $U({\cal B}_0)$.  
If $(x_1^{-})^{r+1} \Psi= 0$ and 
$(x_1^{-})^{r} \Psi \ne 0$ for an integer $r$, 
we say that $x_1^{-}$ is {\it nilpotent of degree $r$} in 
the highest weight representation.  
\end{df} 
 
In a finite-dimensional representation of $U({\cal B}_0)$,  
$x_1^{-}$ is nilpotent, i.e. $(x_1^{-})^s=0$ for some integer $s$.  
For a highest weight vector in a finite-dimensional representation 
of $U({\cal B}_0)$, we can define the highest weight polynomial and 
 highest weight parameters $a_j$  similarly as in the case of 
the $sl_2$ loop algebra \cite{criterion,RAQIS05,DGMTP05}. 

\begin{df}
Let $\Psi$ be a highest weight vector of $U({\cal B}_0)$. 
By applying the Poincar{\' e}-Birkhoff-Witt theorem to $U({\cal B}_0)$, 
it follows that  the highest weight representation 
$U({\cal B}_0)\Psi$ is decomposed into the direct sum of 
subspaces with respect to eigenvalues of $h_0$, and that   
every vector $v$ in the subspace of weight 
$d_0-2n$ is written as follows: 
$$
v = \sum_{1 \le k_1 \le \cdots \le k_n} C_{k_1, \cdots, k_n} x_{k_1}^{-} 
\cdots x_{k_n}^{-} \Psi  \, . 
$$
We call the subspace of weight $d_0-2n$ the sector of degree $n$.  
\end{df}

\begin{prop} 
Let $\Psi$ be a highest weight vector of $U({\cal B}_0)$.  
If $x_1^{-}$ is nilpotent of degree $r$ in the highest weight representation 
$U({\cal B}_0) \Psi$, then 
the sector of degree $2r$ in $U({\cal B}_0) \Psi$ is one-dimensional. 
\label{prop:1-dim} 
\end{prop}

We can show  proposition \ref{prop:1-dim} through the following lemma 
\cite{DGMTP05,criterion}. 
\begin{lem}
Let $\Psi$ be a highest weight vector of $U({\cal B}_0)$. 
We assume that 
$x_1^{-}$ is nilpotent of degree $r$ in $U({\cal B}_0)\Psi$. 
Let us take a non-negative integer $n$ satisfying $n \le r$. Then,   
 for any set of positive integers, $k_1, \ldots, k_n$, we have  
\be 
(x_1^{-})^{r-n} x_{k_1}^{-} \cdots x_{k_n}^{-} \Psi 
= A_{k_1, \ldots, k_n} \, (x_1^{-})^{r} \Psi  \, . 
\label{eq:r-n}
\ee
Here, $A_{k_1, \ldots, k_n}$ is  given by a complex number. 
\end{lem} 

Let us denote by $(X)^{(n)}$ the $n$th power of operator $X$ 
divided by the $n$ factorial, i.e. $(X)^{(n)} = X/n!$.
\begin{lem} 
Let $\Psi$ be a highest weight vector of $U({\cal B}_0)$. If  
$x_1^{-}$ is {\it nilpotent of degree $r$} in $U({\cal B}_0)\Psi$,  then 
$\Psi$ is a simultaneous eigenvector of $(x_0^{+})^{(n)} (x_1^{-})^{(n)}$:   
\be 
(x_0^{+})^{(j)} (x_1^{-})^{(j)} \Omega = \lambda_j \Omega \, , \quad 
\mbox{\rm for} \quad j=1, 2, \ldots, r \, . 
\label{eq:01}
\ee
Here $\lambda_j$ are eigenvalues. 
\label{lem:lambda} 
\end{lem} 
\begin{proof} 
From the Poincar{\' e}-Birkhoff-Witt theorem  
of $U({\cal B}_0)$, it follows that the sector 
of degree 0  in $U({\cal B}_0) \Psi$ is one-dimensional. 
Since  $(x_0^{+})^{(j)} (x_1^{-})^{(j)} \Psi$ is 
in the sector of degree $0$ in $U({\cal B}_0) \Psi$,  
it is proportional to the basis vector $\Psi$. 
\end{proof}

Let $\Psi$ be a highest weight vector in a finite-dimensional 
representation of the Borel subalgebra $U({\cal B}_0)$.  
We now introduce parameters expressing the highest weight of $\Psi$.   
We denote by $\lambda=(\lambda_1, \ldots, \lambda_r)$ 
the sequence of eigenvalues $\lambda_k$ which are defined 
in eq. (\ref{eq:01}). Here we recall that 
in a finite-dimensional representation, $x_1^{-}$ is nilpotent of some degree. 
We define a polynomial $P_{\lambda}(u)$ by the following relation 
\cite{criterion}:  
\be 
P_{\lambda}(u) =\sum_{k=0}^{r} \lambda_k (-u)^k  \, .   
\label{eq:DrinfeldP}
\ee
We call it the {\it highest weight polynomial} of $\Psi$.

Let us factorize polynomial $P_{\lambda}(u)$ as follows 
\be 
P_{\lambda}(u) = \prod_{k=1}^{s} (1 - a_k u)^{m_k} \, ,   
\label{eq:factor}
\ee 
where $a_1, a_2, \ldots, a_s$ are distinct, and their  
 multiplicities are given by  $m_1, m_2, \ldots, m_s$, respectively.   
We denote by ${\bm a}$ the sequence of $s$ 
parameters $a_j$: 
\be 
{\bm a}=(a_1, a_2, \ldots, a_s). 
\ee
Here we note that $r$ is equal to 
the sum of multiplicities $m_j$: $r=m_1 + \cdots + m_s$. 
We define parameters ${\hat a}_i$ for $i=1, 2, \ldots, r$, as follows.  
\be 
{\hat a}_i = a_k \quad {\rm if } \, \, m_1+ m_2 + \cdots + m_{k-1} < i \le  
m_1+ \cdots + m_{k-1} + m_{k} \, . 
\label{eq:hat-a}
\ee
Then, the set $\{ {\hat a}_j \, | j =1, 2, \ldots, r \}$ corresponds to the 
set of parameters $a_j$ with multiplicities $m_j$ for 
 $j=1, 2, \ldots, s$. We denote by ${\hat {\bm a}}$ 
 the sequence of $r$ parameters ${\hat a}_i$: 
\be  
 {\hat {\bm a}}=({\hat a}_1, {\hat a}_2, \ldots, {\hat a}_r) .      
\ee
We call parameters ${\hat a}_i$ 
the {\it highest weight parameters} of $\Psi$. 
It follows from the definition of highest weight polynomial 
${\cal P}_{\lambda}(u)$ given by (\ref{eq:DrinfeldP}) 
and that of highest weight parameters (\ref{eq:factor}) that 
 we have   
\be 
\lambda_{n} = \sum_{1 \le j_1 < \cdots < j_n \le r} 
{\hat a}_{j_1} \cdots {\hat a}_{j_n} \, . 
\label{eq:lambda}
\ee
If the highest weight parameters are nonzero, we define 
${\bar \lambda}_{n}$ for $n=0, 1, \ldots, r$ by 
\be 
{\bar \lambda}_{n} = \sum_{1 \le j_1 < \cdots < j_n \le r} 
{\hat a}_{j_1}^{-1} \cdots {\hat a}_{j_n}^{-1} \, . 
\label{eq:lambda-bar}
\ee

We remark that we may call the highest weight polynomial of $\Psi$ and 
the highest weight parameters of $\Psi$ 
the {\it loop-highest weight polynomial} of $\Psi$
and the {\it loop-highest weight parameters} of $\Psi$, respectively 
 \cite{l-highest}.

%%%%%%%%%%%%%%%%%%%%%%%%%%%%%%%%%%%%%%%%%%%%%%%%%%%%%%%% 
%
\subsection{Borel subalgebra generators with parameters}
%
%%%%%%%%%%%%%%%%%%%%%%%%%%%%%%%%%%%%%%%%%%%%%%%%%%%%%%%

Let ${\bm \alpha}$ denote a finite sequence of complex parameters such as 
${\bm {\alpha}} =(\alpha_1, \alpha_2, \ldots, \alpha_n)$. 
We define generators with $n$ parameters, $x_m^{\pm}({\bm \al})$ and 
$h_m({\bm \al})$, as follows \cite{criterion,RAQIS05,DGMTP05}: 
\begin{eqnarray} 
x_{m}^{\pm}({\bm \al}) & = & 
\sum_{k=0}^{n} (-1)^k x_{m-k}^{\pm} 
\sum_{\{ i_1, \ldots, i_k \} \subset \{1, \ldots, n \}}  
\alpha_{i_1} \alpha_{i_2} \cdots \alpha_{i_k} \, , 
\non \\
h_{m}({\bm \al}) & = & 
\sum_{k=0}^{n} (-1)^k h_{m-k} 
\sum_{\{ i_1, \ldots, i_k \} \subset \{1, \ldots, n \}}  
\alpha_{i_1} \alpha_{i_2} \cdots \alpha_{i_k} \, . 
\label{eq:def-alpha}
\end{eqnarray}
Here, in the case of the Borel subalgebra $U({\cal B}_0)$, 
we define $x_{m}^{+}({\bm \al})$ and  $h_{m}({\bm \al})$ for $m \ge n \ge 0$,  
and $x_{m}^{-}({\bm \al})$ for $m > n \ge 0$.

Let ${\bm \al}$ and ${\bm \beta}$ 
be arbitrary sequences of $n$ and $p$ parameters, respectively.   
Here we have $n, p \ge 0$. 
In terms of generators with parameters 
we express the defining relations of the Borel subalgebra as follows:   
\begin{equation} 
[x_{\ell}^{+}({\bm \al}), x_m^{-}({\bm \beta})] 
= h_{\ell+m}({\bm {\al \beta}}) \, , \quad 
[h_{\ell}({\bm \al}), x^{-}_{m}({\bm \beta})] 
= - 2  x_{\ell+m}^{-}({\bm {\al \beta}}) \, ,  
\label{eq:dfr-AB}
\end{equation}
for $\ell \ge n$ and $m > p$,  and 
\be 
[h_{\ell}({\bm \al}), x^{+}_{m}({\bm \beta})] 
= 2  x_{\ell+m}^{+}({\bm {\al \beta}}) \, . 
\ee
for $\ell \ge n$ and $m \ge p$, 
Here the symbol ${\bm {\al \beta}}$ 
denotes the composite sequence of ${\bm \al}$ and ${\bm \beta}$: 
\be 
{\bm {\al \beta}} 
= (\alpha_1, \alpha_2, \ldots, \alpha_{n}, 
\beta_1, \beta_2, \ldots, \beta_{p}). 
\ee
%
%Using  relations (\ref{eq:dfr-AB}), 
%we can show the following relations for $t \in {\bf Z}_{\ge 0}$: 
%\begin{eqnarray} {[} (x_m^{+}({\bm \al}))^{(t)},  x_{\ell}^{-}({\bm \beta}) {]}%& = & (x_m^{+}({\bm \al}) )^{(t-1)} h_{\ell+m}({\bm \al} {\bm \beta}) 
% + x_{\ell+2m}^{+}({\bm \al} {\bm \beta}{\bm \beta}) 
% (x_{m}^{+}({\bm \al}))^{(t-2)} \, , \non \\ 
%
% {[} x_{\ell}^{+}({\bm \al}), (x_m^{-}({\bm \beta}))^{(t)} {]} 
% & = & (x_m^{-}({\bm \beta}) )^{(t-1)} h_{\ell+m}({\bm \al} {\bm \beta}) 
% - x_{\ell+2m}^{-}({\bm \al} {\bm \beta} {\bm \beta}) 
% (x_{m}^{-}({\bm \beta}))^{(t-2)} \, , \non \\ 
%
% {[} h_{\ell}({\bm \al}), (x_m^{\pm}({\bm \beta}))^{(t)} {]} 
% & = & \pm 2 (x_m^{\pm}({\bm \beta}))^{(t-1)} 
% x_{\ell+m}^{\pm}({\bm \al} {\bm \beta}) \, . \label{eq:AB} \end{eqnarray}

\begin{lem}
Let $\Psi$ be a highest weight vector of $U({\cal B}_0)$. 
If $x_t^{-}({\bm \alpha}) \Psi = 0$ for a positive integer $t$ 
and a sequence of parameters ${\bm \alpha}=(\alpha_1, \ldots, \alpha_n)$ where 
$t > n$,  we have 
$h_{t+m}({\bm \alpha})=0$ and $x_{t+m}^{\pm}({\bm \alpha}) =0$ 
 in the highest weight representation of $\Psi$ for $m \in {\bf Z}_{\ge 0}$: 
for any set of positive integers,  
$k_1, k_2, \ldots, k_n$,  and for $m \in {\bf Z}_{\ge 0}$,  
we have  the following: 
\bea 
& & x_{t+m}^{-}({\bm \alpha}) \, 
x_{k_1}^{-} x_{k_2}^{-}  \cdots x_{k_n}^{-} \Psi = 0 \, , 
\label{eq:x-} \\ 
& & h_{t+m}({\bm \alpha}) \, x_{k_1}^{-} x_{k_2}^{-}  
\cdots x_{k_n}^{-} \Psi = 0 \, , \label{eq:h} \\ 
& & x_{t+m}^{+}({\bm \alpha}) \, 
x_{k_1}^{-} x_{k_2}^{-}  \cdots x_{k_n}^{-} \Psi = 0 \, . \label{eq:x+}  
\eea 
\label{lem:x+hx-}
\end{lem} 
\begin{proof} 
Following  the Poincar{\' e}-Birkhoff-Witt theorem, one can show that 
every vector in the sector of degree $n$ of 
$U({\cal B}_0) \Psi$ is expressed as a linear combination of 
monomial vectors $x_{k_1}^{-} x_{k_2}^{-} \cdots x_{k_n}^{-} \Psi$. 
It is easy to show  (\ref{eq:x-}). 
 By induction on $n$, we can show  (\ref{eq:h}),  and then (\ref{eq:x+}).   
\end{proof}

%%%%%%%%%%%%%%%%%%%%%%%%%%%%%%%%%%%%%%%
%
\subsection{Recurrence relations}
%
%%%%%%%%%%%%%%%%%%%%%%%%%%%%%%%%%%%%%%%

Let us denote by ${\cal B}_{0}^{+}$ 
such a subalgebra of $U({\cal B}_{0})$ 
that is generated by ${x}_{k}^{+}$ for $k \in {\bm Z}_{\ge 0}$. 
Similary as the case of the $sl_2$ loop algebra \cite{criterion}, 
we can show the following: 
\begin{lem} 
The following recurrence relations hold for $n \in {\bm Z}_{\ge 0}$:    
\bea
&({\rm A}_n):&  ({  x}_{0}^{+})^{(n-1)} ({  x}_{1}^{-})^{(n)}  
 =  \sum_{k=1}^{n} (-1)^{k-1} {  x}_{k}^{-} 
({  x}_{0}^{+})^{(n-k)} 
({  x}_{1}^{-})^{(n-k)} \, \mbox{\rm mod} \, 
U({\cal B}_{0}){\cal B}_{0}^{+},   \non \\
&({\rm B}_n):&   
({  x}_{0}^{+})^{(n)} ({  x}_{1}^{-})^{(n)}  
 =  {\frac 1 n} \, \sum_{k=1}^{n} (-1)^{k-1} {  h}_{k} 
 ({  x}_{0}^{+})^{(n-k)} 
 ({  x}_{1}^{-})^{(n-k)} \, \mbox{\rm mod} \,
  U({\cal B}_{0}) {\cal B}_{0}^{+},  \non \\
&({\rm C}_n):&  
  {[} {  h}_j(a),  ({  x}_{0}^{+})^{(m)} 
({  x}_{1}^{-})^{(m)} {]}  = 0 \,  
\mbox{\rm mod} \, U({\cal B}_{0}) {\cal B}_{0}^{+} \, \quad 
\mbox{\rm for} \, \, m \le n \, \,  \mbox{\rm and} \, \,  
j \in {\bm Z}_{\ge 0} \, . \non
\eea
\label{lem:ABC}
\end{lem}

\begin{prop}[Reduction relations]   
\bea 
{x}_{r+1+m}^{-} \, \Psi & = & \sum_{k=1}^{r} (-1)^{r-k} 
\lambda_{r+1-k} \, {x}^{-}_{k+m} \, \Psi \, ,  \quad {\rm for} \, \, 
m \in {\bf Z}_{\ge 0} \, ,  
\label{eq:rr-x} \\ 
{d}_{r+1+m}  & = & \sum_{k=1}^{r} (-1)^{r-k} 
\lambda_{r+1-k} \, {d}_{k+m} \, , \quad {\rm for} \, \, 
m \in {\bf Z}_{\ge 0} \, .  
\label{eq:rr-d}
\eea
\label{lem:rr}
\end{prop} 
\begin{proof}
Reduction relation (\ref{eq:rr-x}) for $m=0$ is derived 
from $({\rm A}_{r+1})$ of lemma \ref{lem:ABC} 
and lemma \ref{lem:lambda}. Applying $h_n$ for $n \ge 0$ 
to reduction relation (\ref{eq:rr-x}) for $m=0$,  
we have reduction relation (\ref{eq:rr-x}) for $m=n$. 
Applying $x_0^{+}$ to (\ref{eq:rr-x}) 
from the left, 
we derive relations (\ref{eq:rr-d}).   
\end{proof} 

\begin{cor}
Let $\Psi$ be a highest weight vector of $U({\cal B}_0)$ and 
${\bm {\hat a}}=({\hat a}_1, \ldots, {\hat a}_r)$ the highest weight 
parameters. In the highest weight representation of $\Psi$
 we have $h_{r+m}({\bm {\hat a}})= 0$,  
$x^{+}_{r+m}({\bm {\hat a}})= 0$ and 
$x^{-}_{r+1+m}({\bm {\hat a}})= 0$ for $m \in {\bm Z}_{\ge 0}$. 
\label{cor:rr}
\end{cor}
\begin{proof}
It follows from lemma \ref{lem:x+hx-} and reduction relations (\ref{eq:rr-x}). 
\end{proof}

\subsection{A theorem on the Borel subalgebra}

We first recall a simple fact in linear algebra. 
Let $x_n$ for $n=0, 1, \ldots,$ be an infinite sequence of numbers 
satisfying a linear recurrence relation:  
\be 
x_{n+r}= \sum_{k=1}^{r} \gamma_k x_{n+r-k} \label{eq:recur}
\ee
We denote by ${\bm x}_n = ^{t}(x_{n+1}, \ldots, x_{n+r})$ 
Then,  for any integer $n$,  
there exists a matrix $A^{[n]}$ such that 
\be 
{\bm x}_n = A^{[n]} {\bm x}_1 . 
\ee
Furtheremore, we have for any $m$ the following: 
\be 
{\bm x}_{n+m} = A^{[n]} {\bm x}_{m+1} . \label{eq:m+n}
\ee

\begin{th1}
Let $\Psi$ be a highest weight vector 
in a finite-dimensional representation of the Borel subalgebra 
$U({\cal B}_0)$.  If all the highest weight parameters 
of $\Psi$, i.e. ${\hat a}_j$,  are nonzero, 
then the action of ${\cal B}_0$ on $\Psi$ can be  extended to 
that of the $sl_2$ loop algebra: Suppose that 
$x_1^{-}$ is nilpotent of degree $r$ 
in the highest weight representation of $\Psi$. 
We define ${\tilde h}_0$ by ${\tilde h}_0= h_0-d_0 +r$, 
and $x_0^{-}$ by 
\be
x_0^{-} = \sum_{j=1}^{r} (-1)^{j-1} {\bar \lambda}_{j} x_j^{-} \, ,   
\ee
where ${\bar \lambda}_{j}$ are given by 
\be 
{\bar \lambda}_{j} = \sum_{i_1 < \cdots < i_j} 
{\hat a}_{i_1}^{-1} \cdots {\hat a}_{i_j}^{-1} \, . 
\ee
We also define  
$x_{-\ell}^{\pm}$ and $h_{-\ell}$ for 
$\ell \in {\bf Z}_{>0}$ 
by  
\bea 
x_{-\ell}^{\pm} & = & \sum_{j=1}^{r} (-1)^{j-1} {\bar \lambda}_{j-1} 
x_{j-\ell}^{\pm} \, ,  \non \\ 
 h_{-\ell} & = & \sum_{j=1}^{r} (-1)^{j-1}  {\bar \lambda}_{j-1} 
h_{j-\ell} \, . 
\eea
Then, they satisfy the defining relations of the $sl_2$ loop algebra 
in the highest weight representation $U({\cal B}_0)\Psi$. 
\label{th:new}
\end{th1}
\begin{proof} 
 If a set of operators $x_j^{+}$, $h_k$ for $j=0, 1, \ldots$, 
 and $x_k^{-}$ for $k=1, 2, \ldots,$ satisfy the 
defining relations of the Borel subalgebra, 
then  the set operators with  $h_0$ being replaced 
by ${\tilde h}_0$ also satisfy the  defining relations (\ref{eq:dfrBorel}).  
We can also show that $x_0^{-}$ satisfies the defining relations 
(\ref{eq:dfrBorel}). Making use of corollary \ref{cor:rr}, 
we can show 
$[x_{m}^{+}, x_{-\ell}^{-}] = h_{m-\ell}$, 
 $[h_m, x_{-\ell}^{\pm}] = (\pm 2) x^{\pm}_{m-\ell}$,   
$[x_{-\ell}^{+}, x_{m}^{-}] = h_{m-\ell}$ 
and $[h_{-\ell}, x_{m}^{\pm}] = (\pm 2) x_{m-\ell}^{\pm}$   
for $m, n \in {\bf Z}_{\ge 0}$. Here we express 
$h_{-\ell}$ and $x_{m}^{\pm}$ in terms of 
linear combinations of $h_{j}$ and $x_{j}^{\pm}$ 
for $j=1, 2, \ldots, r$. 
We calculate commutation relations among $h_{j}$ and $x_{k}^{\pm}$ 
for $j, k =1, 2, \ldots, r$, and then show the defining relations    
through (\ref{eq:m+n}). 
For an illustration, we show 
 $[h_{-\ell}, x_{m}^{\pm}]= (\pm2) x^{\pm}_{-\ell+m}$ as follows.  
\bea 
[h_{-\ell}, x_{m}^{\pm}] & = & 
[ \left(A^{[-\ell-1]} {\bm h}_{0}\right)_1, x_m^{\pm}] \non \\
& = &   \sum_{j=1}^{r} A^{[-\ell-1]}_{1,j} [ h_{j}, x_m^{\pm}] \non \\ 
& = &   (\pm2) \sum_{j=1}^{r} A^{[-\ell-1]}_{1,j} \, x_{j+m}^{\pm} \non \\ 
& = & (\pm2) \left( A^{[-\ell-1]} {\bm x}^{\pm}_{m} \right)_1 \non \\ 
& = & (\pm2) x^{\pm}_{-\ell+m} \, . \non 
\eea
Similary, we can show $[x_{-\ell}^{+}, x_{m}^{-}]= h_{-\ell+m}$.   
Furthermore, we can show 
  $[x_{-m}^{+}, x_{-\ell}^{-}] = h_{-m-\ell}$ and 
 $[h_{-m}, x_{-\ell}^{\pm}] = (\pm 2) x^{\pm}_{-m-\ell}$ 
for $m, n \in {\bf Z}_{\ge 0}$.
\end{proof} 

We note that in a finite-dimensional 
representation of the Borel subalgebra, 
the highest weight $d_0$ is not necessarily 
given by an integer.

We should note that Benkart and Terwilliger (2004) 
have shown that an irreducible finite-dimensional 
representation of the Borel subalgebra is extended uniquely 
to an irreducible representation of the $sl_2$ loop algebra 
\cite{Benkart}. Thus, if the highest weight representation is 
irreducible, then theorem \ref{th:new} should be equivalent to the 
result \cite{Benkart}. Here we 
 recall that an irreducibility criterion is known 
for a finite-dimensional highest weight representation 
of the Borel algebra with nonzero  highest weight parameters
as follows \cite{DGMTP05}: 
\begin{prop}
Let $\Psi$ be a highest weight vector 
in a finite-dimensional representation of $U({\cal B}_0)$. 
We denote by ${\hat a}_j$ the highest weight parameters. 
It generates an irreducible representation 
if and only we have   
\be 
\sum_{j=0}^{s} (-1)^{s-j} \mu_{s-j} x_{j+1}^{-} \Psi = 0  \, ,  
\label{eq:irrep}
\ee
where $\mu_{k}$ $(k=1, 2, \ldots, s)$ are given by 
$$ 
\mu_{k} = \sum_{1 \le i_1 < \cdots < i_k \le s} a_{i_1} \cdots a_{i_k}  
\, . 
$$
\end{prop}
Thus, in a finite-dimensional highest weight representation of 
$U({\cal B}_0)$, 
if the highest weight vector  satisfies the condition (\ref{eq:irrep}), 
it is irreducible and we can also show by making 
use of the result of Ref. \cite{Benkart} 
that it is extended into a finite-dimensional 
highest weight representation  of the $sl_2$ loop algebra.

%%%%%%%%%%%%%%%%%%%%%%%%%%%%%%%%%%%%%%%%%%%%%%%%%%%% 
\section{Application to the twisted XXZ spin chain} 
\subsection{Regular Bethe vectors under the twisted B.C. as highest weight of $U(L(sl_2))$}  

We briefly discuss some important points of the infinite-dimensional 
symmetry of the twisted XXZ spin chain at roots of unity. 
Some details will be given elsewhere. 

For the periodic XXZ spin chain at a root of unity $q_0$ with 
$q_0^{2N}=1$, Fabricius and McCoy     
conjectured \cite{Odyssey} 
that every Bethe state should be a highest weight vector of the 
$sl_2$ loop algebra. Then, it has been 
explicitly proved 
for some sectors of $S^Z$ mod $N$ \cite{RegularXXZ} 
that every regular Bethe state is a highest weight vector 
of the $sl_2$ loop algebra. 
For the twisted XXZ spin chain, we can also show 
in some sectors of $S^Z$ mod $N$ that  
  every regular  
Bethe state is highest weight with respect to the Borel subalgebra 
$U({\cal B}_0)$.

Let us denote by $t_1, t_2, \ldots, t_R$ 
solutions of  the twisted BA equations 
\be 
{\frac {a_{6V}(t_{j})} {d_{6V}(t_j)} } = q^{- 2\varphi} 
\prod_{k=1, k \ne j}^{R} 
 {\frac {f(t_{k}-t_{j}) } {f(t_{j}-t_{k})}}  \, , 
\quad {\rm for} \, j= 1, 2, \cdots, R \, .  
\label{eq:BAE}
\ee
Here  $a_{6V}(z)$ and $d_{6V}(z)$ are given by 
\be 
a_{6V}(z)= \sinh^{L}(z+ \eta) \, , \quad d_{6V}(z) = \sinh^{L}(z-\eta) \, . 
\ee
We recall that function $f(z-w)$ is given by  
$$
f(z-w)= {\frac {\sinh(z-w-2 \eta)} {\sinh(z-w)}} \, ,   
$$
and $q = \exp(2 \eta)$. 

\begin{df}
Let  $t_1, t_2, \ldots, t_R$ satisfy  
the twisted Bethe ansatz equations (\ref{eq:BAE}). 
We call them {\it Bethe roots}. 
We call a set of Bethe roots,  $t_1, t_2, \ldots, t_R$,    
{\it regular},  if they are finite and distinct. 
In terms of a set of regular Bethe roots, 
 $t_1, t_2, \ldots, t_R$,  
we define the regular Bethe state $|R \rangle$ by 
\be 
| R \rangle = B(t_1) B(t_2) \cdots B(t_R) \, | 0 \rangle 
\ee 
Here $| 0 \rangle$ denotes the vacuum state in which all spins are up. 
\end{df}

We recall the following conjecture \cite{RegularXXZ}.  
\begin{conj} 
For the twisted Bethe ansatz equations (\ref{eq:BAE}) 
at a root of unity $q_0$, 
every set of regular Bethe roots 
 ${\tilde t}_1, {\tilde t}_2, \ldots, {\tilde t}_R$
 gives an isolated solution of eqs. (\ref{eq:BAE}).   
\label{conj:main}
\end{conj} 
We can show conjecture \ref{conj:main} for some particular cases. 
For instance, it is trivial for such Bethe states with  one down-spin.  

Assuming conjecture \ref{conj:main} we can show the following theorem. 
\begin{th1} 
Let $q_0$ be a root of unity with $q_0^{2N}=1$ for some integer $N$,  
and $\varphi$ an integer or a half-integer such that 
we have $q_0^{2N \varphi}=1$. 
Every regular Bethe state $|R \rangle$ 
gives a highest weight vector of the Borel subalgebra $U({\bm B}_0)$ 
if it is in sector A: $S^Z \equiv + \varphi$ (mod $N$) 
or $S^Z \equiv - \varphi$ (mod $N$) where $q_0^{2N}=1$,  
or in sector B: $S^Z \equiv N/2 + \varphi$ (mod $N$) or  
$S^Z \equiv N/2 - \varphi$  (mod $N$) where $q_0^{N}=1$ with $N$ odd. 
\label{th:main} 
\end{th1}
The proof of theorem \ref{th:main} will be given in a different paper.

%%%%%%%%%%%%%%%%%%%%%%%%%%%%%%%%%%%%%%%%%%%%%%%%
\subsection{Derivation of the degree of nilpotency for  $x_1^{-}$}

Let us assume that a regular Bethe state $|R \rangle$ 
of the twisted XXZ spin chain is 
in such a sector of $S^Z$ where theorem \ref{th:main} holds, 
i.e.  $|R \rangle$ is a highest weight vector of $U({\cal B}_0)$. 
For a highest weight representation 
of $U({\cal B}_0)$ generated by $| R \rangle$, 
it follows from the finite dimensionality 
that  $x_1^{-}$ is nilpotent. 

Let ${\tilde t}_1, {\tilde t}_2, \ldots, {\tilde t}_R$ be 
such a set of regular Bethe roots at a root of unity $q_0$ with 
$q_0^{2N}=1$ that leads to the regular Bethe state $|R \rangle$. 
Let us define  $\eta_0$ by $q=\exp 2\eta_0$. 
We now introduce the following function: 
\be 
Y(v; \varphi) = \sum_{\ell=0}^{N-1} {\frac {q_0^{-\varphi (2 \ell +1)} 
\prod_{j=1}^{L} \sinh(v-(2 \ell +1)\eta_0)}  
{\prod_{j=1}^{R} (\sinh(v-{\tilde t}_j - 2\ell \eta_0) 
\sinh(v-{\tilde t}_j - 2(\ell+1) \eta_0)}} 
\ee 
It follows from the twisted Bethe ansatz equations (\ref{eq:BAE})
that $Y(v; \varphi)$ is a Laurent polynomial of $\exp 2Nv$.  
Here we recall that when $\varphi=0$ it is nothing but the polynomial 
introduced by Fabricius and McCoy for the XXZ spin chain at roots of unity 
under the periodic boundary conditions \cite{Odyssey,RegularXXZ}. 

Making use of the Laurent polynomial $Y(v; \varphi)$, 
we can show the following.  
\begin{prop}
Let $| R \rangle$ be a regular Bethe state 
of the twisted XXZ spin chain at $q_0$ a root of unity with 
$q_0^{2N}=1$ and in a sector of $S^Z$ and with twist parameter 
$\varphi$ such that the conditions of theorem \ref{th:main} hold.  
Then, it is a  highest weight vector of $U({\cal B}_0)$, and     
all the highest weight parameters of $| R \rangle$ are nonzero. 
\label{prop:nonzero}
\end{prop} 
 
Let us consider a regular Bethe state 
$| R \rangle$ in such a sector of $S^Z$ where  
we have $S^Z \pm \varphi \equiv 0$ (mod $N$).  
It follows from theorem \ref{th:new} and 
proposition \ref{prop:nonzero} that the 
highest weight representation generated by any given 
regular Bethe state in a sector of $S^Z \pm \varphi \equiv 0$ (mod $N$)  
extends to a highest weight representation of the $sl_2$ loop algebra.  
Thus,  the Borel subalgebra symmetry 
of the twisted XXZ spin chain is extended to 
the $sl_2$ loop algebra symmetry.

\subsection{Conjecture of the $sl_2$ loop algebra symmetry of the twisted transfer matrix}
 
We now  present the following conjecture: 
\begin{conj}
The twisted XXZ spin chain 
at roots of unity with twist parameter $\varphi$ being integers 
should have the $sl_2$ loop algebra symmetry in every sector 
of $S^Z$ mod $N$. 
\end{conj}

\section{Representations of the Onsager algebra derived from highest weight 
representations of $U(L(sl_2))$}

The Onsager algebra is generated by operators 
$A_{m}$ and $G_{\ell}$ ($\ell, m=0, \pm 1, \pm2, \ldots$)  
satisfyng the  following defining relations 
\cite{OnsagerI,Dolan-Grady,Davies,Roan,Uglov-Ivanov,Date-Roan}: 
\bea 
[A_{\ell}, A_{m} ] & = & 4 G_{\ell-m} \, , \non \\ 
{[} G_{\ell}, A_m {]} & = & 2 A_{m+\ell} - 2 A_{m - \ell} \, , \non \\  
{[} G_{\ell}, G_m {]} & = & 0  \, . \label{eq:OA} 
\eea
We remark that Davies has shown that 
if generators $A_n$ satisfy a linear recurrence relation 
\be 
\sum_{k=-n}^{n} \gamma_k A_{k-n} = 0 \,, 
\ee
then they are expressed in terms of the generators of $sl_2$  as follows 
\cite{Davies}: 
\bea 
A_m & = & 2 \sum_{j=1}^{n}  
\left(e_j^{+} \otimes z_j^m + e_j^{-} \otimes z_j^{-m} \right) \, ,
 \, \non \\
G_k & = & \sum_{j=1}^{n}  
\left( h_j \otimes z_j^k - h_j \otimes z_j^{-k} \right)
\eea
Here $e_j^{\pm}$ and $h_k$ satisfy 
\be 
[h_j, e_k^{\pm}] = \pm 2 e_j^{\pm} \, \delta_{jk} \, , \quad 
[e_j^{+}, e_k^{-}] = h_j \, \delta_{j,k}
\ee

Let ${\bm \alpha}$ denote a finite sequence of complex parameters such as 
${\bm {\alpha}} =(\alpha_1, \alpha_2, \ldots, \alpha_{\ell})$. 
Similarly as (\ref{eq:def-alpha}), 
we define generators with $\ell$ parameters, $A_m({\bm \al})$ and 
$G_m({\bm \al})$, as follows: 
\begin{eqnarray} 
A_{m}({\bm \al}) & = & 
\sum_{k=0}^{\ell} (-1)^k A_{m-k} 
\sum_{\{ i_1, \ldots, i_k \} \subset \{1, \ldots, \ell \}}  
\alpha_{i_1} \alpha_{i_2} \cdots \alpha_{i_k} \, , 
\non \\
G_{m}({\bm \al}) & = & 
\sum_{k=0}^{\ell} (-1)^k G_{m-k} 
\sum_{\{ i_1, \ldots, i_k \} \subset \{1, \ldots, \ell \}}  
\alpha_{i_1} \alpha_{i_2} \cdots \alpha_{i_k} \, . 
\label{eq:def-alpha-AG}
\end{eqnarray}

Let $\Omega$ be a highest weight vector in a finite-dimensional 
representation of the $sl_2$ loop algebra, $U(L(sl_2))$. 
We define operators $A_m$ and $G_k$ in terms of 
generators $x_m^{\pm}$ and $h_k$ of $U(L(sl_2))$ by 
\be 
 A_m = x_m^{+} + x_m^{-} \, , \quad G_k= h_k - h_{-k} \, . 
\ee
Then, operators $A_m$ and $G_k$ 
satisfy the defining relations of the Onsager algebra. 
Furthermore, we can show recurrence relations of $A_m$'s.  
\begin{prop} 
Let $\Omega$ be a highest weight vector in a finite-dimensional 
representation of the $sl_2$ loop algebra, $U(L(sl_2))$. 
If generators $x_m^{-}$ of $U(L(sl_2))$ satisfy a recurrence relation,  
$x_n^{-}({\bm \beta})\Omega = 0$, for 
a sequence of nonzero parameters ${\bm \beta}=(\beta_1, \beta_2, \ldots, 
\beta_n)$, then we have 
\be 
A_n({\bm \beta}; {\bm {\bar \beta}}) = 0 . 
\ee
Here ${\bm {\bar \beta}}$ denotes ${\bm {\bar \beta}}
 = (\beta_1^{-1}, \beta_2^{-1}, \ldots, \beta_n^{-1})$ . 
\label{prop:recur}
\end{prop} 

In Ref. \cite{ND}, the $sl_2$ loop algebra symmetry is derived for 
a spin-$N/2$ fusion model of the six-vertex model 
at $q_0$ being an $N$th root of unity,  
which is associated with the superintegrable chiral Potts model.   
From the representations of the $sl_2$ loop algebra 
derived from the fusion model, we can thus construct 
representations of the Onsager algebra. Then, through proposition 
\ref{prop:recur}, we derive recurrence relations for generators 
of the Onsager algebra. It should thus be an interesting problem to   
discuss connections to the 
Onsager algebra symmetry of the $Z_N$ symmetric Hamiltonian 
given by von Gehlen and Rittenberg. 
We shall discuss them elsewhere. 

\section*{Acknowledgement} 
The author would like to thank the organizers of 
the workshop RAQIS, Sep. 11-14, 2007, LAPTH, Annecy France,  
for giving him the opportunity 
to participate it. He would also like to thank 
Prof. M. Jimbo, Prof. B.M. McCoy and Dr. A. Nishino for helpful comments. 
He is quite grateful to Prof. P. Baseilhac for useful discussion on 
the Onsager algebra during the stay in Tours in September 2006. 
This work is partially supported by 
Grant-in-Aid for Scientific Research (C) No. 17540351.


\begin{thebibliography}{99}  


\bibitem{Barber}
F.C. Alcaraz, M. Barber and M. Batchelor, 
Conformal invariance, the XXZ chain and the operator content 
of two-dimensional critical systems 
Ann. Phys., NY {\bf 182} (1988) 280-343. 


\bibitem{Alcaraz} F.C. Alcaraz, U. Grimm and V. Rittenberg, 
The XXZ Heisenberg chain, conformal 
invariance and the operator content of 
$c<1$ systems, Nucl. Phys. B {\bf 316} (1989) 
735-768.   

\bibitem{SS}
B.S. Shastry and B. Sutherland, 
Twisted boundary conditions and effective mass in Heisenberg-Ising 
 and Hubbard rings 
Phys. Rev. Lett. {\bf 65}, 243 (1990). 

\bibitem{ACTWu} V.E. Korepin and A.C.T. Wu, 
Int. J. Mod. Phys. {\bf B5}, 497 (1991).  

\bibitem{Fowler} N. Yu and M. Fowler, 
Twisted boundary conditions and the adiabatic ground state 
for the attractive XXZ Luttinger liquid, 
 Phys. Rev. B {\bf 46}, 14583 (1992). 
  

\bibitem{DFM} T. Deguchi, K. Fabricius and B. M. McCoy,
The $sl_2$ loop algebra symmetry of the six-vertex model at roots of unity, 
%\par \noindent 
J. Stat. Phys. {\bf 102} (2001) 701-736. 
%

\bibitem{FM1} K. Fabricius and B.~M. McCoy, 
Bethe's Equation Is Incomplete for the XXZ Model at Roots of Unity,    
J. Stat. Phys. {\bf 103}(2001) 647-678.   

\bibitem{FM2} K. Fabricius and B.~M. McCoy, 
Completing Bethe's Equations at Roots of Unity,  
J. Stat. Phys. {\bf 104}(2001) 573-587.  


\bibitem{Odyssey} K. Fabricius and B. M. McCoy, 
Evaluation Parameters and Bethe Roots 
for the Six-Vertex Model at Roots of Unity, 
Progress in Math. Phys. {\bf 23} 
 ({\it MathPhys Odyssey 2001}), 
edited by  M. Kashiwara and T. Miwa, 
(Birkh{\"a}user, Boston, 2002) 119-144. 
% 


\bibitem{criterion} T. Deguchi, 
Irreducibility criterion for a finite-dimensional 
highest weight representation of the $sl_2$ loop algebra and 
the dimensions of reducible representations, 
J. Stat. Mech. (2007) P05007. 


\bibitem{RegularXXZ} T. Deguchi, 
XXZ Bethe states as highest weight vectors 
of the $sl_2$ loop algebra at roots of unity, 
J. Phys. A {\bf 40} (2007) 7473-7508.  
% (cond-mat/0503564).
% 

\bibitem{twisted} T. Deguchi, 
The $sl_2$ loop algebra symmetry of the twisted transfer matrix 
of the six-vertex model at roots of unity, 
J. Phys. A: Math. Gen. {\bf 37} (2004) 347-358.  
%
%(cond-mat/0306498)
%(in the proceedings of RAQIS03, Annecy, France, March 2003.)   

\bibitem{Korff} C. Korff, The twisted XXZ chain at roots of unity 
revisited, 
J. Phys. A {\bf 37} (2004) 1681-1689. 
%


\bibitem{Benkart} G. Benkart and P. Terwilliger, 
Irreducible modules for the quantum affine algebra $U_q({\hat sl}_2)$
and its Borel subalgebra, 
J. Algebra {\bf 282} (2004) 172-194. 




\bibitem{ND} 
A. Nishino and T. Deguchi, 
The $L(sl_2)$ symmetry of the Bazhanov-Stroganov model 
associated with the superintegrable chiral Potts model,  
 Phys. Let. A {\bf 356} (2006) pp. 366-370.
% (cond-mat/0605551) 
%

\bibitem{Helen} H. Au-Yang and J.H.H. Perk, 
Eigenvectors in the superintegrable Model, 
arXiv:0710.5257 [math-ph]. 


\bibitem{RAQIS05} T. Deguchi, 
The Six-Vertex Model at Roots of Unity 
and some Highest Weight Representations of 
the $sl_2$ Loop Algebra,  
%cond-mat/0603112, 
Ann. Henri Poincar{\'e} {\bf 7} (2006), 1531-1540. 
%(Birkh{\" a}user Verlag, Basel/Switzerland)   


\bibitem{DGMTP05} T. Deguchi,  
Generalized Drinfeld polynomials for highest weight vectors 
of the Borel subalgebra of the $sl_2$ loop algebra,  
``Differential Geometry and Physics'',  
the Proc. of the 23rd International Conference 
of Differential Geometric Methods 
in Theoretical Physics, Tianjin, China, 20-26 August 2005, 
eds. M.-L. Ge and W. Zhang (Chern Institute of Mathematics, Tianjin, 
China,  2006) pp. 169-178. 
%(math-ph/0606071) 


\bibitem{l-highest} It was pointed out by Prof. M. Jimbo 
during the workshop RAQIS07 that technical terms such as 
loop-highest weight modules are employed by H. Nakajima 
(see for instance, arXiv:math/0103008v2 [math.QA]).   


\bibitem{OnsagerI} L. Onsager, 
Crystal Statistics. I. A Two-Dimensional Model with 
an Order-Disorder Transition, Phys. Rev. {\bf 65} (1944) 117-149. 


\bibitem{Dolan-Grady} L. Dolan and M. Grady, 
Conserved charges from self-duality, Phys. Rev. D {\bf 25} (1982) 
1587-1604. 

\bibitem{Davies} B. Davies, Onsager's algebra and superintegrability, 
J. Phys. A: Math. Gen. {\bf 23} (1990) 2245-2261.  


\bibitem{Roan} S. Roan, Onsager's algebra, loop algebra 
and chiral Potts model, Bonn preprint 1991.

\bibitem{Uglov-Ivanov} D.B. Uglov and I.T. Ivanov, $sl(N)$  Onsager's algebra 
and Integrability, J. Stat. Phys. {\bf 112} (1996) 87-113. 

\bibitem{Date-Roan} E. Date and S. Roan, The structure of the Onsager 
algebra by closed ideals, J. Phys. A: Math. Gen. {\bf 33} (2000) 
3275-3296. 


\end{thebibliography}
\end{document}